\pdfoutput=1

\newif\ifDraft
\newif\ifSubmission
\newif\ifProceeding
\newif\ifFull
\newif\ifArxiv

\Draftfalse
\Submissionfalse
\Proceedingfalse
\Arxivtrue
\Fullfalse
\documentclass{llncs}
\usepackage{hyperref,mathptmx,graphicx,amssymb,subfig,amsmath}
\usepackage{compress}
\usepackage{color}

\let\oldendproof\endproof
\def\endproof{\qed\oldendproof}

\newif\ifUseAppendix
\ifSubmission
\UseAppendixtrue
\else
\UseAppendixfalse
\fi

\newif\ifInAppendix
\newif\ifDeferToAppendix

\ifUseAppendix
\InAppendixfalse
\DeferToAppendixtrue  % Easier to recognize - defer these parts to Appendix
\else
\InAppendixfalse
\DeferToAppendixfalse % Same reason...
\fi

\ifDraft
\newcounter{mycommentcounter}

\newcommand{\Comment}[2][Comment]{\refstepcounter{mycommentcounter}%
    \ifhmode%
     \unskip%
     {\dimen1=\baselineskip \divide\dimen1 by 2 %
       \raise\dimen1\llap{\tiny
	{-\themycommentcounter-}}}\fi%
     \marginpar[{\renewcommand{\baselinestretch}{0.8}%
       \hspace*{-2em}\begin{minipage}{12em}\footnotesize%
[\themycommentcounter]:%
\raggedright \underline{#1}: #2\end{minipage}}]{\renewcommand{\baselinestretch}{0.8}%
       \begin{minipage}{12em}\footnotesize%
[\themycommentcounter]: \raggedright%
\underline{#1}: #2\end{minipage}}%
}
\else
\newcommand{\Comment}[2][Comment]{\relax}        % use this for commenting 
\fi

\newcommand{\highlight}[1]{{\bfseries\itshape #1}}

\title{Lombardi Drawings of Graphs}

\author{Christian A. Duncan\inst{1} \and
  David Eppstein\inst{2} \and
  Michael T. Goodrich\inst{2} \and \\
  Stephen G. Kobourov\inst{3} \and
  Martin N\"ollenburg\inst{2}}

\institute{\noindent%meaningless macro to break out of vertical mode so \inst works
\inst{1}Department of Computer Science, Louisiana Tech. Univ., Ruston, Louisiana, USA\\
\inst{2}Department of Computer Science, University of California, Irvine, California, USA\\
\inst{3}Department of Computer Science, University of Arizona, Tucson, Arizona, USA}

\begin{document}

\maketitle

\begin{abstract}
We introduce the notion of Lombardi graph drawings, named after the American abstract artist Mark Lombardi. 
In these drawings, edges are represented as circular arcs rather than 
as line segments or polylines, and the vertices have 
\emph{perfect angular resolution}: the edges are equally spaced 
around each vertex. 
We describe algorithms for finding Lombardi drawings of regular graphs, 
graphs of bounded degeneracy, and certain families of planar graphs.
\end{abstract}

%\ifProceeding
% Save some space and bring Figure 1 to front page.
\vspace{-1cm}
%\fi

\section{Introduction}
\label{sec:intro}
The American artist Mark Lombardi~\cite{lh-mlgn-03} was famous for his drawings of social networks representing conspiracy theories. 
% In contrast to the polyline style used in many graph drawing papers, 
Lombardi used curved arcs to represent edges, leading to a strong aesthetic quality and high readability.
Inspired by this work,
we introduce the notion of a \emph{Lombardi drawing} of a graph, in which 
edges are drawn as circular arcs with \emph{perfect angular resolution}: consecutive edges are evenly spaced around each vertex.
While not all vertices have perfect angular resolution 
in Lombardi's work, the even spacing of edges 
around vertices is clearly one of his 
aesthetic criteria; see Fig.~\ref{fig:real}.

Traditional graph drawing methods rarely guarantee perfect 
angular resolution, but poor edge distribution can nevertheless 
lead to unreadable drawings.
Additionally, while some tools provide options to draw edges 
as curves, most rely on straight-line edges, and 
it is known that maintaining good angular resolution can result 
in exponential drawing area for straight-line drawings of 
planar graphs~\cite{gt-pdara-94,mp-arpg-94}. 
Our requirement of perfect angular resolution forces us to use curved edges, 
since even very simple graphs such as 
cycle graphs cannot be drawn with perfect 
angular resolution and straight edges.

\begin{figure}[hbt!]
  \centering
  \includegraphics[width=4.6in]{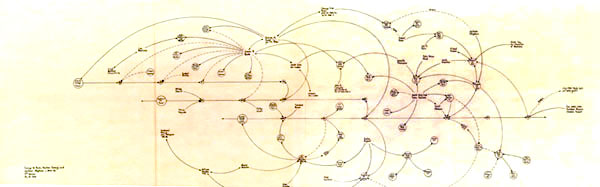}
  \vspace*{-6pt}
  \caption{Mark Lombardi, \textit{George W. Bush, Harken Energy, 
  and Jackson Stevens c.1979-90}, 1999. 
  Graphite on paper, $20 \times 44$ inches~\cite{lh-mlgn-03}.}
  \label{fig:real}
\end{figure}

\ifFull
\subsection{New Results}
\else
\paragraph{New Results.}
\fi
We define a \highlight{Lombardi drawing} of a graph $G$ to be a drawing of 
$G$ in the plane in which vertices are represented as points (or as disks or labels centered on those points), edges are represented as line segments or circular arcs between their endpoints, and every vertex has perfect angular resolution, as measured by the angle formed by the tangents to the edges at the vertex.
We do not necessarily insist that the drawings are free of crossings; the drawings of Lombardi had crossings, sometimes even in cases where they could have been avoided. We also do not consider crossings when we measure the angular resolution of a drawing.
However, we do require that the only vertices that intersect the arc for an edge $(u,v)$ are its two endpoints $u$ and $v$.

Several of Mark Lombardi's drawings used a circle as their overall shape.
We define a \highlight{circular Lombardi drawing} to be a Lombardi
drawing in which the vertices lie on a circle.
\ifFull
It is almost equivalent to ask for a Lombardi drawing in which the vertices lie on a straight line, as circles and straight lines can be transformed into each other (preserving circularity of arcs and local angular resolution) by a M\"obius transformation; the only difference is that vertices on a circle can be connected by a cycle of edges that lie entirely on the circle while vertices on a line cannot.
\fi
Similarly, we define a \highlight{$k$-circular Lombardi drawing} to be
a Lombardi drawing in which the vertices lie on
$k$ concentric circles.
\ifFull
As can be seen from Fig.~\ref{fig:real}, Mark Lombardi often
used the $x$-coordinates of vertices to convey extra information
such as a timeline, so circular Lombardi drawings (transformed to straighten the circle containing the vertices) may be of interest in graph drawing applications in which an additional dimension such as time is to be visualized.
\fi
We provide the following:
\begin{itemize}
\ifProceeding
% Save some space
\setlength{\itemsep}{0cm}
\fi
\item We characterize the regular graphs that have circular Lombardi drawings, and we find efficient algorithms for constructing these drawings.
\item We describe methods of finding Lombardi drawings for any $2$-degenerate graph (a graph that may be reduced to the empty graph by repeated removal of vertices of degree at most $2$) and many but not all $3$-degenerate graphs.
\item We investigate the graphs that have planar Lombardi drawings. We show that certain subclasses of the planar graphs always have such drawings, but that there exist planar graphs with no planar Lombardi drawing.
\item We implement an algorithm for constructing $k$-circular Lombardi drawings and use it to draw many symmetric graphs.
\end{itemize}

\ifFull
\subsection{Related Work}
\else
\paragraph{Related Work.}
\fi
Although most previous work on angular resolution concerns straight-line 
drawings (e.g., see~\cite{dv-aptg-96,gt-pdara-94,mp-arpg-94})
or polyline drawings (e.g., see~\cite{gm-ppdga-98,k-dpguc-96}), 
the angular resolution of drawings with circular-arc edges was previously studied by Cheng {\em et al.}~\cite{cdgk-dpg-01}, who showed that maintaining bounded angular resolution in planar drawings may require exponential area even with circular-arc edges.
Our circular Lombardi drawings use a circular layout of vertices that is 
already popular (e.g., see~\cite{bb-crcl-05,gk-icl-07,st-fcdn-99}). 
However, previous methods for circular layouts draw edges 
as straight line segments or curves perpendicular to the circle, 
neither of which leads to good angular resolution.

Efrat {\em et al.}~\cite{eek-flc-07} show that given a fixed placement of the vertices of a planar graph, determining whether the edges can be drawn with circular arcs so that there are no crossings is NP-Complete. 
% They also show that if the choices for circular arcs are exactly the two possible half-circles, then the problem has an efficient polynomial-time algorithm.
For fixed position drawings with cubic B\'ezier curves, 
Brandes {\em et al.}~\cite{bs-adctd-07,DBLP:journals/jgaa/BrandesW00} use
force-directed algorithms to maximize the angular resolution 
and Brandes, Shubina, and Tamassia~\cite{bst-iarvgn-00} 
rotate optimal angular resolution templates.
Aicholzer {\em et al.}~\cite{aaadj-at-10} show that, for a given embedded planar triangulation with fixed vertex positions, one can find a circular-arc 
drawing of the triangulation that maximizes the minimum angular resolution 
by solving a linear program.
Finkel and Tamassia~\cite{ft-cgduf-04} also 
try to optimize angular resolution using force-directed methods 
for laying out graphs with curved edges.
Di~Battista
and Vismara~\cite{dv-aptg-96} 
give a nonlinear optimization characterization that
can find straight-line drawings of embedded planar graphs
with a prescribed assignment of angles if such drawings exist.

Any tree may be drawn with straight edges and perfect angular resolution. 
However, in a separate paper~\cite{degkn-dtwpa-10}, 
we show that (when the order of the edges is fixed around each vertex) 
straight-line tree drawings with perfect angular resolution 
may require exponential area, whereas Lombardi drawings can achieve polynomial area.

\ifProceeding
% Save some space and bring a figure closer inline to its reference.
\vspace{-0.3cm}
\fi

\section{Circular Lombardi Drawings of Regular Graphs}
\label{sec:regular}

We begin by investigating \emph{circular Lombardi drawings}, Lombardi drawings in which all vertices are placed on a circle. As we show, drawings of this type exist for many regular graphs.
Our proofs use the following basic geometric observation:
\ifProceeding
% Save some space and bring a figure closer inline to its reference.
\vspace{-1pt}
\fi
\begin{property}
\label{prop:circleAngle}
Let $A$ be a circular arc or line segment connecting two points $p$ and $q$ that both lie on circle $O$. Then $A$ makes the same angle to $O$ at $p$ that it makes at $q$.
Moreover, for any $p$ and $q$ on $O$ and any angle $0 \le\theta\le \pi$, there exists an arc, line segment, or pair of collinear rays $A$ connecting $p$ and $q$, making angle $\theta$ with $O$, and lying either inside or outside of $O$.
\end{property}
\ifProceeding
% Save some space and bring a figure closer inline to its reference.
\vspace{-2pt}
\fi
The case of two collinear rays is problematic (we only allow edges to be represented by arcs or line segments) but easily avoided by perturbing the vertices on~$O$.

\ifProceeding
% Save space...
\vspace{-5pt}
\fi
\begin{lemma}
\label{lem:regular}
A $d$-regular graph $G$ has a circular Lombardi drawing if and only if $G$ can be decomposed into a disjoint union of $1$-regular and $2$-regular graphs and one of the following conditions is true:
$d\not\equiv 2\pmod{4}$,
one of the $2$-regular subgraphs is bipartite, or
one of the $2$-regular subgraphs is a Hamiltonian cycle.
\end{lemma}
\ifProceeding
% Save space...
\vspace{-11pt}
\fi
\begin{proof}
Suppose $G$ has a circular Lombardi drawing on a circle $O$ centered at~$o$; in this drawing, define the \emph{twist} $\theta_v$ of a vertex $v$ to be the sharpest of the angles between line segment $vo$ and the edges incident to $v$ (with positive sign if one of the edges forming the sharpest angle is clockwise of the line segment, and negative sign if there is only one edge forming the sharpest angle and it is counterclockwise of $v$). Then if $v$ and $w$ are adjacent in $G$, $\theta_v=-\theta_w$ except when there are two equal sharpest angles at both $v$ and $w$, in which case $\theta_v=\theta_w$. In each connected component either all vertices have the same twist, and have edge angles that are symmetric with respect to reflections through axis $vo$, or the component is bipartite, all vertices on one side of the bipartition have one twist, and all vertices on the other side of the bipartition have the opposite twist.

We can decompose each connected component of $G$ into $1$-regular and $2$-regular graphs by partitioning the edges of the component according to the angle they make with circle $O$.
For a bipartite component in which the vertices on the two sides of the bipartition have different twists, this forms a decomposition into $1$-regular graphs (some of which may be combined in pairs to form bipartite $2$-regular graphs). When $d$ is $2$ mod $4$ and a component of $G$ is not bipartite, the only possibilities for a symmetric twist are to make some edges parallel or perpendicular to $O$. Edges that are parallel to $O$ must be drawn as arcs of $O$ through all vertices, so they form a Hamiltonian cycle. Edges perpendicular to $O$ must form even-length cycles that alternate between the inside and outside of~$O$. Thus, in all cases a graph with a circular Lombardi drawing can be decomposed into $1$-regular and $2$-regular graphs matching the conditions of the lemma.

In the other direction, suppose that $G$ can be decomposed into $1$-regular and $2$-regular graphs with the additional conditions of the lemma. By combining pairs of $1$-regular graphs into a single $2$-regular graph, we may assume that all but at most one of these subgraphs are $2$-regular. Then we may choose an evenly spaced set of angles, draw each $2$-regular graph as a set of arcs that meet $O$ at one of these fixed angles, and draw the $1$-regular graph (if it exists) as a set of arcs that are perpendicular to and interior to $O$.
If $d$ is divisible by four, we can choose these angles in such a way that no angle is parallel to the circle $O$ and no angle is perpendicular to $O$. If $d$ is odd, the angles can be chosen so that the $1$-regular subgraph of $G$ is perpendicular to and interior to $O$, and all other angles are neither perpendicular nor parallel to $O$. If $d$ is congruent to 2 mod 4 and one of the $2$-regular graphs is a Hamiltonian cycle, we may draw it using edges that lie on $C$, placing the vertices in the order of this cycle. And if $d$ is congruent to 2 mod 4 and one of the $2$-regular graphs is bipartite, we may draw it using edges that are perpendicular to $O$, taking care in the vertex placement to avoid using an edge that connects two diametrally opposite points on $O$ via an exterior arc. In both of these cases where $d$ is 2 mod 4 we then draw the other subgraphs of the decomposition using arcs that are neither parallel to nor perpendicular to $O$.
\end{proof}

\ifProceeding
% Save space...
\vspace{-12pt}
\fi
\begin{theorem}
\label{thm:regular}
Every regular graph $G$ of degree divisible by four has a circular Lombardi drawing. A regular graph of odd degree has a circular Lombardi drawing if and only if it has a perfect matching. A regular graph of degree congruent to two modulo four has a circular Lombardi drawing if and only if it is Hamiltonian or has a $2$-regular bipartite subgraph.
In the cases of odd degree and degree divisible by four, when a circular Lombardi drawing exists it can be constructed in polynomial time.
\end{theorem}
\ifProceeding
% Save space...
\vspace{-12pt}
\fi
\begin{proof}
This follows from Lemma~\ref{lem:regular} together with Petersen's theorem that a regular graph of even degree can always be decomposed into 2-regular subgraphs~\cite{m-jptrg-92,p-trg-91}.
\end{proof}

Testing for the existence of a $2$-regular bipartite subgraph in a regular graph is NP-complete (in $3$-regular graphs, it is equivalent to $3$-edge-coloring) but we have not determined its complexity for the case of interest to us, $d$-regular graphs in which $d$ is congruent to two modulo four.

Figures~\ref{fig:regular}(a--c) show drawings produced by this method for $3$-regular, $4$-regular, and $6$-regular graphs. Figure~\ref{fig:regular}(d) shows a 3-regular graph that does not have a perfect matching, and that therefore has no circular Lombardi drawing.

\begin{figure}
\begin{center}
\vspace*{-8pt}
\subfloat[]{\includegraphics[height=0.8in]{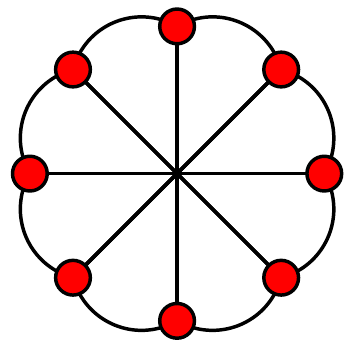}}
\hfill 
\subfloat[]{\includegraphics[height=1in]{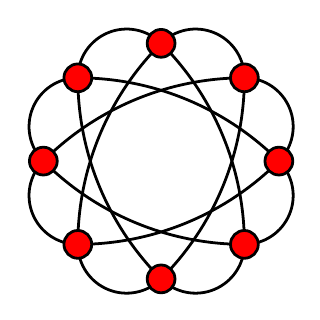}}
\hfill
\subfloat[]{\includegraphics[height=1.4in]{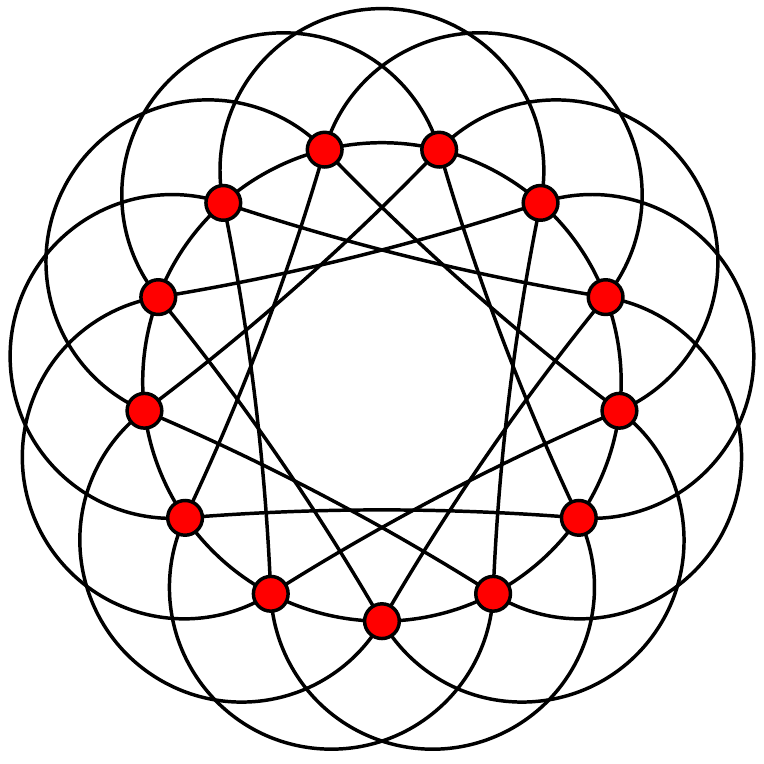}}
\hfill
\subfloat[]{\includegraphics[height=1.2in]{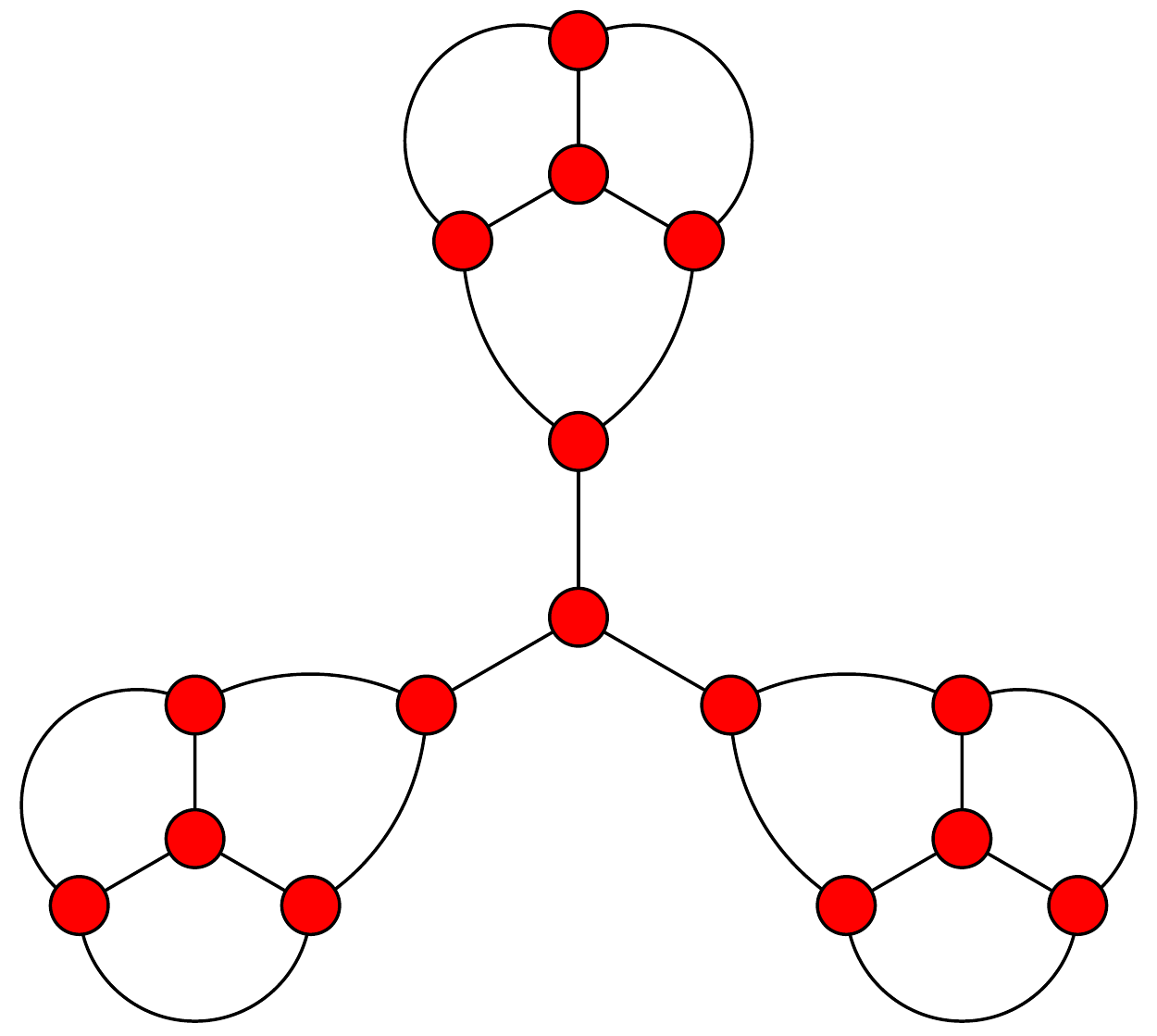}}
\end{center}
\vspace*{-12pt}
\caption{
(a) A circular Lombardi drawing of the 3-regular Wagner graph;
(b) A circular Lombardi drawing of the 4-regular graph $K_{4,4}$;
(c) The 6-regular Paley graph connecting integers modulo 13 if their difference is a quadratic residue;
(d) A 3-regular graph that has no perfect matching and therefore has no circular Lombardi drawing.}
\label{fig:regular}
\end{figure}

For bipartite regular graphs of bounded degree the method of Theorem~\ref{thm:regular} again leads to a linear-time algorithm.

\begin{corollary}
Every bipartite $d$-regular graph has a circular Lombardi drawing that can be constructed in time $O(dn\log d)$.
\end{corollary}

\begin{proof}
It is known that every bipartite regular graph can be decomposed into perfect matchings in the given time bound~\cite{a-saecbm-03,cos-ecbm-01,s-becdmt-99}.\footnote{The fact that every regular bipartite graph has a decomposition into matchings is commonly attributed to K{\"o}nig~\cite{k-gem-31}, but is equivalent to a result proved in terms of point-line configurations in the 1894 Ph.D. thesis of Ernst Steinitz.} The result follows by applying Theorem~\ref{thm:regular} to this decomposition.
\end{proof}

\begin{corollary}
Every $d$-regular graph for which $d$ is a power of two, with the exception of $2$-regular non-bipartite disconnected graphs, has a circular Lombardi drawing that can be constructed in time $O(dn\log d)$.
\end{corollary}

\begin{proof}
Repeatedly decompose the graph into pairs of subgraphs with half the degree by taking alternating edges of an Euler tour~\cite{g-uepecbm-76} and then apply Theorem~\ref{thm:regular} to the decomposition.
\end{proof}

\begin{corollary}
Every $3$-regular bridgeless graph has a circular Lombardi drawing that can be constructed in time $O(n\log^3 n\log\log n)$.
\end{corollary}

\begin{proof}
The result that every $3$-regular bridgeless graph has a perfect matching (equivalently, a decomposition into a $2$-regular and a $1$-regular subgraph) is known as Petersen's theorem~\cite{p-trg-91}. Such a matching can be found in the stated time bound via an algorithm based on dynamic 2-edge-connectivity testing data structures~\cite{bbdl-eapmt-01,hlt-pldfda-01,t-nofdgc-00}.
\end{proof}

\section{Two-Degenerate and Three-Degenerate Graphs}

The \highlight{degeneracy} of a graph $G$ is the minimum number $d$ such that $G$ can be reduced to the empty graph by repeatedly removing a vertex of degree at most $d$; equivalently, it is the minimum degree in the subgraph of $G$ that maximizes the minimum degree~\cite{lw-kdg-70}. If a graph $G$ has degeneracy at most $d$, it is known as $d$-degenerate. In this section we consider algorithms for drawing $2$-degenerate and $3$-degenerate graphs,  with a specified cyclic ordering
of the edges around each vertex. The main idea of these algorithms is to delete a low-degree vertex, draw the remaining graph with the appropriate angles at each of its vertices, and then find a position for the deleted vertex that allows it to be connected to the drawing of the remaining graph.

For 2-degenerate graphs, when we add
back the vertices in reverse order of deletion, there is always a
circle on which they can be added so we can
choose one point on the circle that is not crossed by a previously
drawn feature. For 3-degenerate graphs there are two points at which
the point can be added to give the correct edge angles (the common intersection points of three circles) so there might be circumstances under which this addition is forced to create an undesirable edge-vertex or vertex-vertex intersection.

\ifInAppendix
\label{sec:prelims}
% Skip much of this - except for the proof
\else
%\subsection{Preliminaries}

\ifDeferToAppendix
% Defer the proof to the appendix...
The results in this section rely on the following geometric property; 
we defer the proof to the appendix.
\else
\ifProceeding
% Defer the proof to the full version of the paper...
The results in this section rely on the following geometric property,
which is proven in the full version of the paper~\cite{arxiv-version}.
\else
% Proof is not deferred - given in place...
%    E.g. in Arxiv (or journal) we include the proof explicitly
The results in this section rely on the following geometric property:
\fi
\fi

\begin{property}
\label{prop:arcLocus}
Suppose we are given two points $p$ and $q$ with associated
vectors $\vec{v_p}$ and $\vec{v_q}$ and an angle $\theta_{pq}$.
Consider all pairs of circular arcs 
that leave $p$ and $q$ with tangent vectors $\vec{v_p}$ and $\vec{v_q}$
respectively and meet at an angle $\theta_{pq}$.
The locus of meeting points for these pairs of arcs is a circle.
\end{property}
\fi

\ifInAppendix
We now present the proof for Property~\ref{prop:arcLocus}.
\fi

\ifDeferToAppendix
% Defer the proof to the appendix.
\else
\ifProceeding
% Defer the proof to the full paper
\else
\ifUseAppendix
\begin{proof}[Of Property~\ref{prop:arcLocus}]
\else
\begin{proof}
\fi
Let $r_1$ be the meeting point of one such pair of arcs.
Let $O$ be the circle defined by the three points $p$, $q$, and $r_1$.
From Property~\ref{prop:circleAngle}, the angle $\theta_p$
that the arc from $p$ makes with $O$ as it leaves $p$
is the same as when it arrives at $r_1$.
Similarly, let $\theta_q$ be the angle of the arc with $O$ at both
$q$ and $r_1$.
Therefore, we know that the angle formed by the intersection of 
the two arcs at $r_1$ is
$\theta_{pq} = \pi - \theta_p - \theta_q$;
see Fig.~\ref{fig:circle}(a).

\begin{figure}
\begin{center}
  \subfloat[]{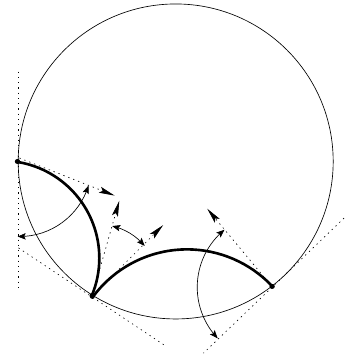}
  \hfil
  \subfloat[]{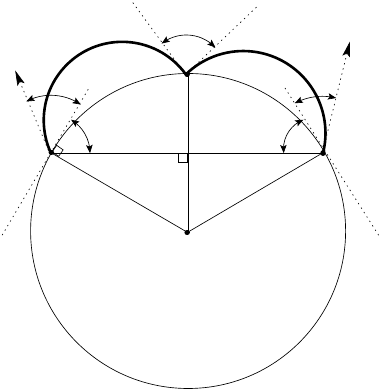}
\end{center}
\vspace{-.5cm}
 \caption{\sf (a) Angle calculation; (b) Circle construction.}
  \label{fig:circle}
\end{figure}

Now, for any other point $r_2$ on $O$,
a circular arc from $p$ through $r_2$ with the same
outgoing tangent vector $\vec{v_p}$ must again form the same
angle $\theta_p$ with $O$ at both $p$ and $r_2$.
The same holds for the angle $\theta_q$ at $q$ and $r_2$.
Therefore, the angle formed by the intersection of the two arcs
at $r_2$ is also $\theta_{pq}$.

We can also determine the equation 
for this circle $O$.
Our goal is to calculate the angle formed by the center of $O$
and the two points $p$ and $q$.
From that, we can use basic trigonometry to 
calculate the position of the center based
on the positions of $p$ and $q$.
For simplicity, assume that the two fixed points $p$ and $q$ are 
horizontally aligned; see Fig.~\ref{fig:circle}(b). 
Let $r$ be the point on $O$ halfway between $p$ and $q$.
Since $r$ lies directly above the center of the circle,
we know that the desired angle is exactly $2x$,
where $x$ is the angle formed by the horizontal line (from $p$ to $q$) 
and the tangent to $O$ at $p$ (or $q$).
From $\vec{v_p}$, we know the angle, say $\theta_{ph}$,
between the outgoing arc from $p$ 
and the horizontal line.
In Fig.~\ref{fig:circle}(b), this corresponds to the angle $\theta_p + x$.
Similarly we have angle $\theta_{qh} = \theta_q + x$.

Finally, from above, we know the angle at $r$ is 
$\theta_{pq} = \pi - \theta_p - \theta_q$.
Solving for $x$, yields that 
$2x=\theta_{ph} + \theta_{qh} - \theta_p - \theta_q = 
  \theta_{ph} + \theta_{qh} + \theta_{pq} - \pi$.
\end{proof}
\fi
\fi

\subsection{2-Degenerate Graphs}

\begin{theorem}
\label{thm:2degenerate}
Every 2-degenerate graph with a specified cyclic ordering of
the edges around each vertex has a Lombardi drawing.
\end{theorem}
\begin{proof}
Order the vertices by repeatedly removing
a low-degree vertex.
Reinsert the vertices in reverse order creating subgraphs $G_0,
G_1 \dots G_n$ with the invariant
that after each insertion the drawing is a {\em partial} Lombardi
drawing $\Gamma_i$ of $G_i$ 
where some vertices may not yet have all of their 
neighbors placed.
To insert a new vertex $v=v_{i+1}$ with degree two in $G_{i+1}$ 
(the case for degree one is simpler)
let $p$ and $q$ be its two neighbors in $G_{i+1}$.
Since there is a specified ordering around $p$,
which has already been placed in $\Gamma_i$,
there is a unique tangent vector $\vec{v_p}$ associated with the arc 
from $p$ to $v$.
Similarly, there is a unique tangent vector $\vec{v_q}$.
In addition, since the degree of $v$ in $G$ is known and the 
ordering of the neighbors at $v$ is also given, there is a unique
angle $\theta_{pq}$ associated with the two arcs from $p$ and $q$
to $v$.
From Property~\ref{prop:arcLocus}, we may choose to place $v$
at any position on the defined circle.
Choosing a point $v$ that does not coincide with any
other arcs or vertices already placed guarantees we have a
valid drawing $\Gamma_{i+1}$.
\end{proof}

\begin{corollary}
\label{cor:outerplanar}
Every outerplanar or series-parallel graph has a Lombardi drawing.
\end{corollary}
\begin{proof}
This follows from the fact that these graphs
are 2-degenerate.
\end{proof}

\subsection{3-Degenerate Graphs}

An algorithm following the same approach can be used to draw many, but not all, 3-degenerate graphs.
In this case we have three points $p$, $q$, and $r$
that we want to connect by arcs to an unplaced new vertex $v$.
Each pair of known points yields a circle of possible choices
for $v$.
These three circles, $O_{pq}, O_{pr}, O_{qr}$, have to pairwise
cross, and where they cross the third one must also cross
because fixing the angles between two pairs of incoming arcs
at the new point fixes all angles. Every graph with maximum degree four is either 4-regular or 3-degenerate, so the same algorithm applies in this case.

However, for certain graphs and certain orderings of the edges around the vertices of the graph, this algorithm can fail by placing a vertex on another edge or vertex. An example in which this occurs is the seven-vertex split graph $G_7$ formed by adding four independent vertices $p$, $q$, $r$, and $s$ to a triangle $xyz$, with an edge from each of $p$, $q$, $r$, and $s$ to each of $x$, $y$, and $z$, as shown in Figure~\ref{fig:imposs3degen}. In any Lombardi drawing of $G_7$ with the edge order as shown, we can assume by making an appropriate M\"obius transformation of the drawing that $xyz$ is equilateral. It follows that the only possible locations for $p$, $q$, $r$, and $s$ are the centroid of the equilateral triangle and the point at infinity, so at least two vertices would have to be placed at the same point, forming an invalid drawing.

\begin{figure}
  \centering
  \vspace*{-12pt}
  \strut\hfill
  \subfloat[]{\includegraphics[scale=0.35,page=1]{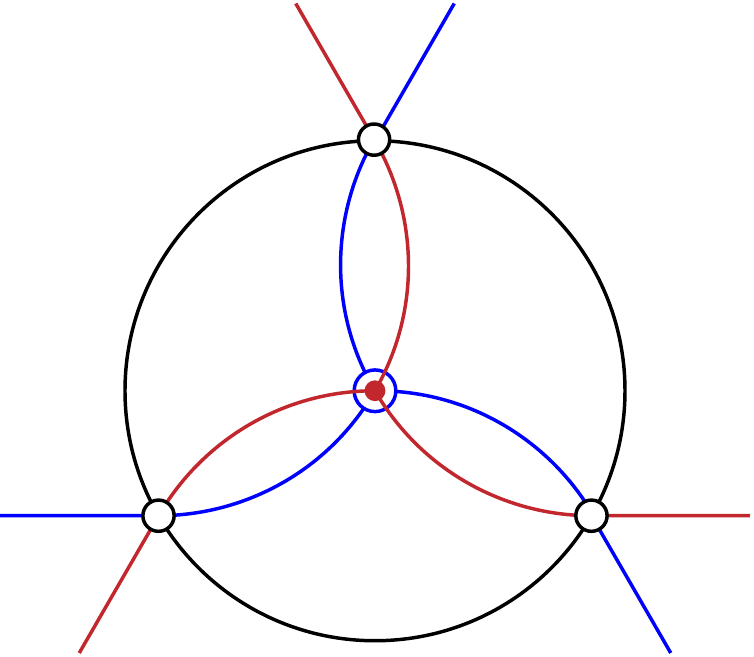}}
  \hfill
  \subfloat[]{\includegraphics[scale=0.35,page=2]{imposs3degen}}
  \hfill
  \subfloat[]{\includegraphics[scale=0.35,page=3]{imposs3degen}}
  \hfill\strut
  \vspace*{-8pt}
  \caption{A 7-vertex 3-degenerate graph that has no Lombardi drawing with 
  the given vertex ordering.
  (a) A M\"obius transformation makes one triangle equilateral, 
  forcing the other 4 vertices to be placed at the centroid and the point 
  at infinity; 
  (b) A different transformation with finite vertex locations;
  (c) A straight-line drawing of the graph.}
  \label{fig:imposs3degen}
  \vspace*{-6pt}
\end{figure}

\section{Non-Crossing Lombardi Drawings}
\ifProceeding
% Save some space and bring a figure closer inline to its reference.
\vspace*{-1pt}
\fi
\subsection{Planar Graphs Without Planar Lombardi Drawings}

Not every planar graph has a planar Lombardi drawing. To see this, consider the \emph{$k$-nested triangle graphs}, maximal planar graphs with $3k$ vertices formed by $k$ nested triangles with $k-1$ six-cycles connecting consecutive triangles. A $k$-nested triangle graph may also be formed geometrically by gluing $k-1$ octahedra end-to-end.

As can be seen in Figure~\ref{fig:triangles}, the 2-nested and 3-nested triangle graphs have planar Lombardi drawings.
The 4-nested triangle graph, however, does not. If it did have such a drawing, its middle two triangles would form circles (the only smooth curve formed by three circular arcs). By an appropriate M\"obius transformation, the outer circle $O$ can be assumed to have its three vertices equally spaced around it. The three cirles $C_1$, $C_2$, and $C_3$ that (by Property~\ref{prop:arcLocus}) describe the potential positions of the vertices on the inner circle have the same radius as $O$ and meet at the center of $O$, and the inner circle would have to be tangent to all three of $C_1$, $C_2$, and $C_3$. However, the only circle tangent to all three is exterior to $O$, concentric with $O$ and having twice the radius of $O$. Therefore, using an edge ordering around each vertex that comes from a planar embedding but enforcing perfect angular resolution leads to a nonplanar drawing, shown in Figure~\ref{fig:triangles}(c).

\begin{figure}
  \vspace*{-18pt}
  \centering
  \subfloat[$k=2$]{\includegraphics[width=2.5cm]{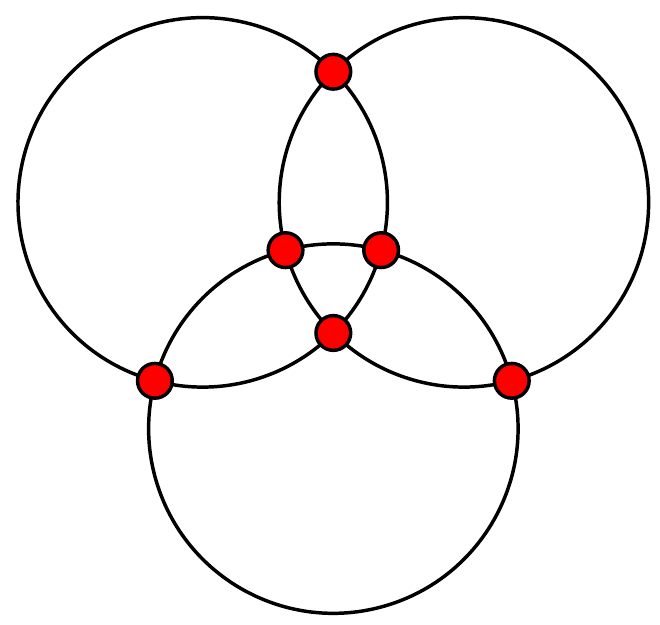}}
  \hfil
  \subfloat[$k=3$]{\includegraphics[width=4.5cm]{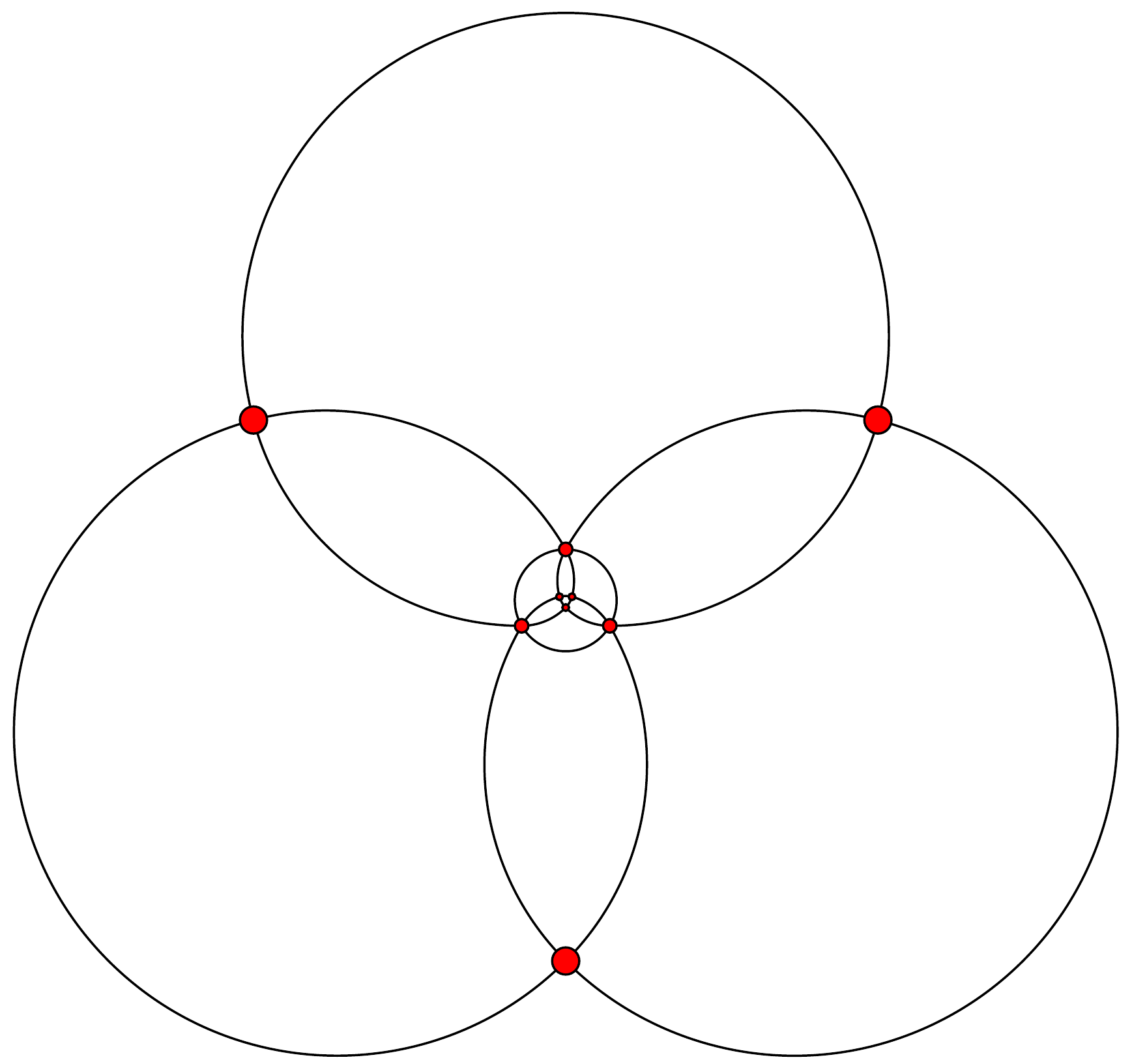}}
  \hfil
  \subfloat[$k=4$]{\includegraphics[width=4.5cm]{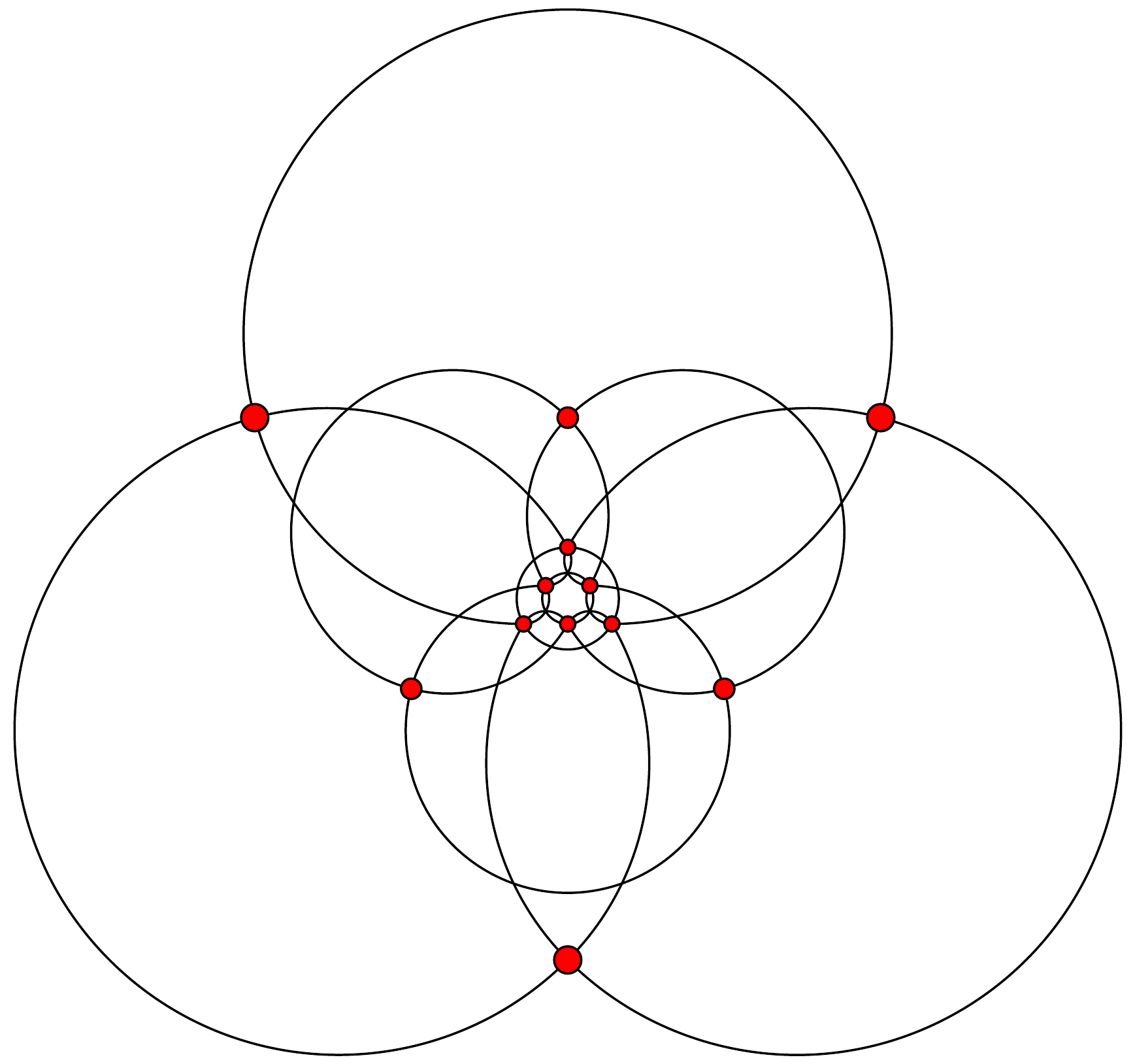}}
  \vspace*{-8pt}
  \caption{\sf $k$-nested triangle graphs. The $2$-nested and $3$-nested triangle graphs have planar Lombardi drawings, but the $4$-nested triangle graph does not.}
  \label{fig:triangles}
\end{figure}

\subsection{Halin Graphs}

A Halin graph~\cite{h-uszbg-64} is a planar graph obtained from a plane tree $T$ (with at least four vertices and with no vertices of degree 2), by connecting all the leaves of $T$ into a cycle in the order given by its embedding. As we now describe, Halin graphs (and the graphs formed in the same way from trees with degree-2 vertices) have planar Lombardi drawings that can be constructed using hyperbolic geometry.

We draw $T$ within a Poincar\'e disk model of the hyperbolic plane, with its leaves on the boundary circle of the model, and then draw the cycle connecting the leaves outside this circle. If $T$ is drawn using hyperbolic line segments, with perfect angular resolution, then its edges will form circular arcs in the Poincar\'e model; the conformal (angle-preserving) nature of the Poincar\'e model implies that the angular resolution of the hyperbolic line segments equals the angular resolution of these Euclidean arcs.

For a given straight-line drawing of a rooted tree in the hyperbolic plane, and a non-root vertex $v$, partition the hyperbolic plane into wedges bounded by the bisectors of the angles around the parent of $v$ and define the \highlight{dominance region} of $v$ to be the wedge containing $v$. Equivalently, in a Voronoi diagram generated by the rays from the parent of $v$ to its children, the dominance region of $v$ is the Voronoi cell containing $v$. We define a \highlight{good hyperbolic drawing} of a rooted tree $T$ to be a drawing in which the edges are straight line segments or rays in the hyperbolic plane, the leaves are placed on the circle at infinity, and the dominance regions for two vertices $v$ and $w$ are either nested within each other (if one of the two vertices is an ancestor of the other) or disjoint otherwise. Two dominance regions in a good hyperbolic drawing are shown in Figure~\ref{fig:halin}(a).

\begin{figure}[tb]
  \centering
  \subfloat[]{\includegraphics[height=3cm,page=1]{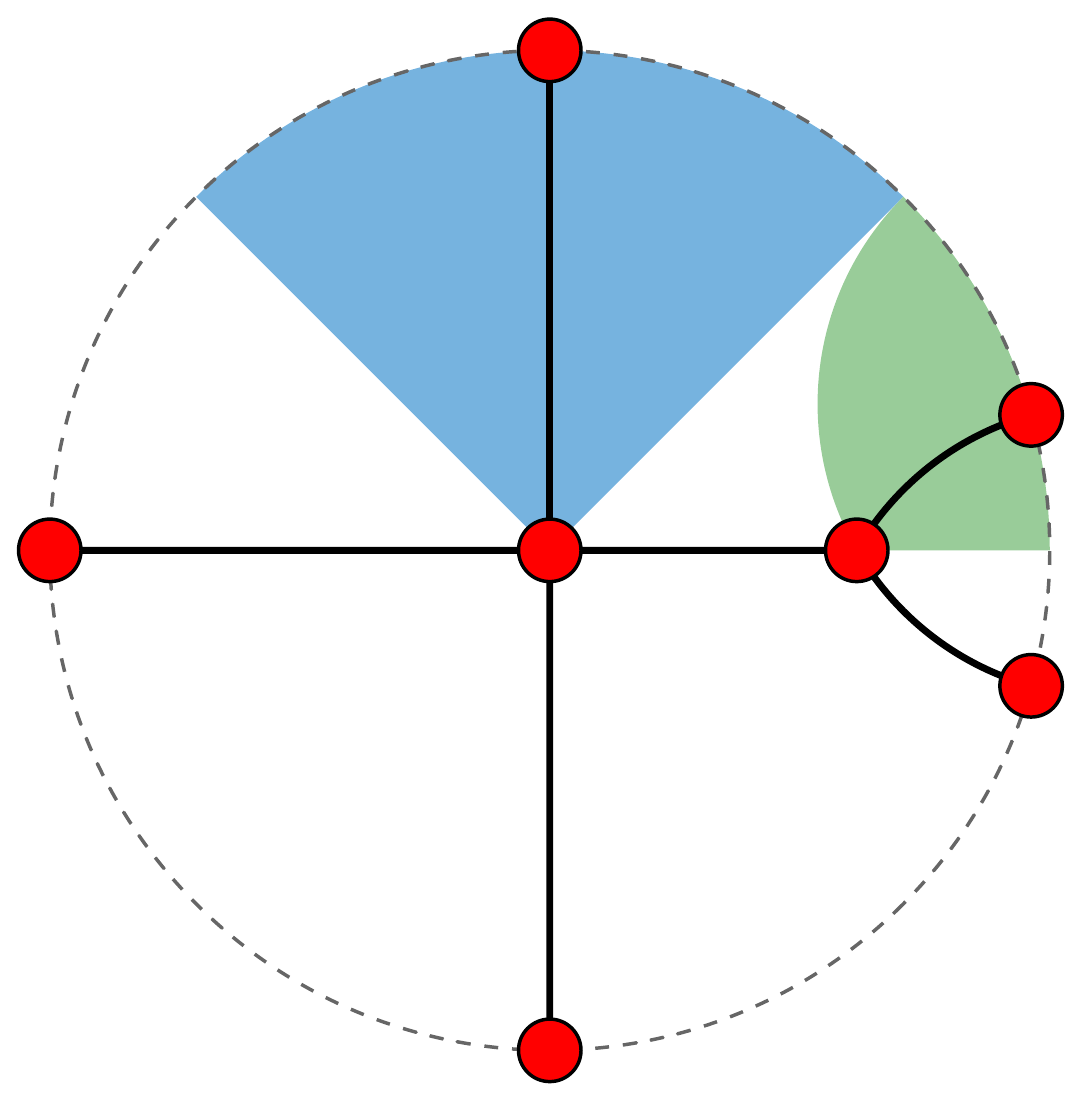}}
  \hfil
  \subfloat[]{\includegraphics[height=3cm,page=2]{halin.pdf}}
  \caption{(a) A good hyperbolic drawing of a seven-node tree, with the dominance regions of two leaves of the tree shown as shaded regions; (b) The Lombardi drawing formed by adding arcs
outside the Poincar\'e model, at $30^\circ$ angles to the boundary, connecting consecutive leaves.}
  \label{fig:halin}
\end{figure}

\begin{lemma}
\label{lem:hypertree}
Every rooted tree has a good hyperbolic drawing.
\end{lemma}

\begin{proof}
We use induction on the number of non-leaf nodes in the given tree $T$. As a base case, when there is one non-leaf node, it may be placed at the center of the Poincar\'e disk model of the hyperbolic plane with its leaves at the limit points of equally-spaced rays (radii of the disk model).
Otherwise, let $v$ be a non-leaf that is as far from the root of $T$ as possible, and let $T'$ be formed from $T$ by removing all children of $v$. Then by induction, $T'$ has a good hyperbolic drawing. In this drawing, $v$ is on the circle at infinity; let $R$ be the ray connecting the parent of $v$ to $v$. For any position $x$ along this ray, let $\theta_x$ be the maximum angle made to $R$ by a line that stays within the dominance region of $v$. Then $\theta_x$ varies continuously along $R$, starting from a value of $\pi/d$ at the parent of $v$ (where $d$ is the degree of the parent) and ending with a value of $\pi$ at $v$ itself. If the degree of $v$ in $T$ is $d'$, there must be an intermediate position $x$ on $R$ for which $\theta_x=\pi(1-1/d')$. If we move $v$ to $x$ and place its leaf children at the limit points of equally spaced rays around $x$, the result is a good hyperbolic drawing of $T$.
\end{proof}

\begin{theorem}
Every Halin graph has a planar Lombardi drawing that may be constructed in linear time.
\end{theorem}

\begin{proof}
Root the tree $T$ at an arbitrarily chosen non-leaf node, and construct a good hyperbolic drawing of $T$ according to Lemma~\ref{lem:hypertree}. Draw the cycle connecting the leaves of $T$ using circular arcs that meet the circle bounding the Poincar\'e model at angles of $30^\circ$ as in Figure~\ref{fig:halin}(b). Then each non-leaf node of $T$ has perfect angular resolution from the tree drawing, and each leaf node has perfect angular resolution because the ray connecting it to its parent in $T$ is perpendicular to the boundary circle and therefore at $120^\circ$ angles from the two arcs connecting it to adjacent leaves.
\end{proof}

%\subsection{Outerpillars}
%We might be able to do outerpillars...

\subsection{Other Classes of Planar Graphs}

The networks formed by two-dimensional soap bubbles naturally form 3-regular planar Lombardi drawings: they have circular arcs as their edges (the boundaries between bubbles), and $120^\circ$ angles at each vertex where three arcs meet~\cite{m-sbr2s-94}. However, we do not have a precise characterization of the graphs that can be formed in this way.

The vertices of every Platonic solid, Archimedean solid, and prism lie on a common sphere. In all but two cases (the snub cube and snub dodecahedron) one may draw the edges of the polyhedron as circular arcs on the sphere with perfect angular resolution. By stereographic projection, each of these graphs has a Lombardi drawing in the plane. For instance, Figure~\ref{fig:triangles}(a) depicts the graph of the octahedron drawn in this way.

All outerplanar and series-parallel graphs have Lombardi drawings (Corollary~
\ref{cor:outerplanar}), but we do not
know whether they all have planar Lombardi drawings.

\section{The Lombardi Spirograph}

We have implemented a program for constructing $k$-circular Lombardi drawings of graphs with dihedral symmetry; we call it the \emph{Lombardi Spirograph}, as its drawings resemble those created by the Spirograph$^{\mathrm{TM}}$ drawing toy produced by Hasbro, Inc. Our program places vertices on $k$ concentric circles; the input specifies not only the number of vertices per circle and the set of edges to be drawn, but also the order in which those edges are incident at each vertex. Each vertex can have at most three neighbors on smaller circles; a circle on which the vertices have two or three inward neighbors has a unique radius for which the vertices have perfect angular resolution, whereas the radius for circles on which the vertices have one inner neighbor is chosen heuristically.

Figures \ref{fig:regular} (a--c), \ref{fig:triangles} (a~\&~b), and~\ref{fig:spirograph} were all drawn using this program.

\begin{figure}[p]
  \centering
  \subfloat[Petersen graph]{\includegraphics[width=3.5cm]{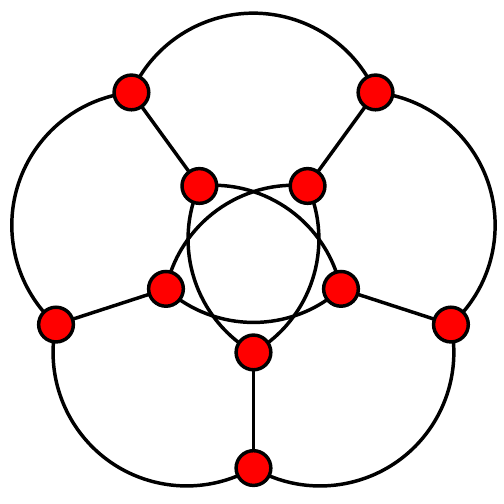}}
  \hfil
  \subfloat[$K_6$]{\includegraphics[width=3.5cm]{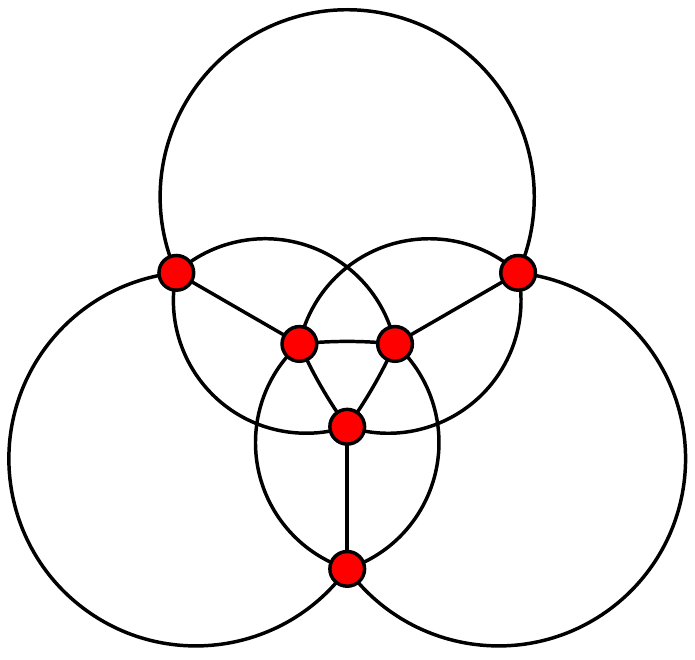}}
  \hfil
  \subfloat[Gr\"otzsch graph]{\includegraphics[width=3.5cm]{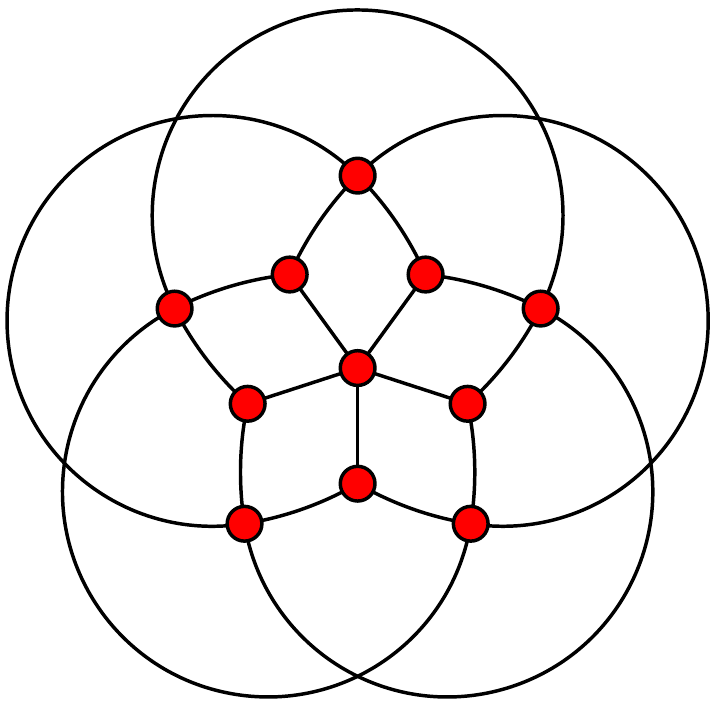}}\\
  \subfloat[Nauru graph $G(12,5)$]{\includegraphics[width=5.75cm]{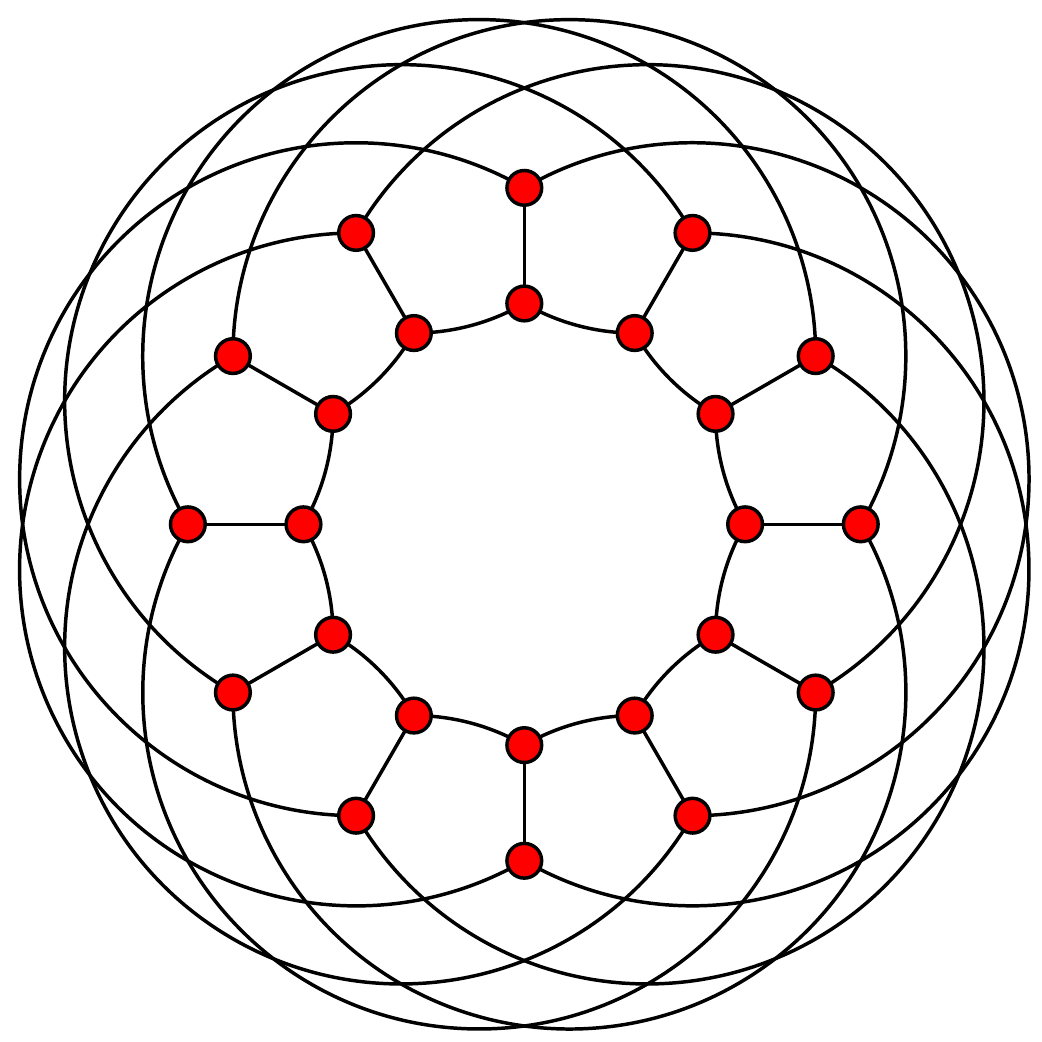}}
  \hfil
  \subfloat[Brinkmann graph]{\includegraphics[width=5.75cm]{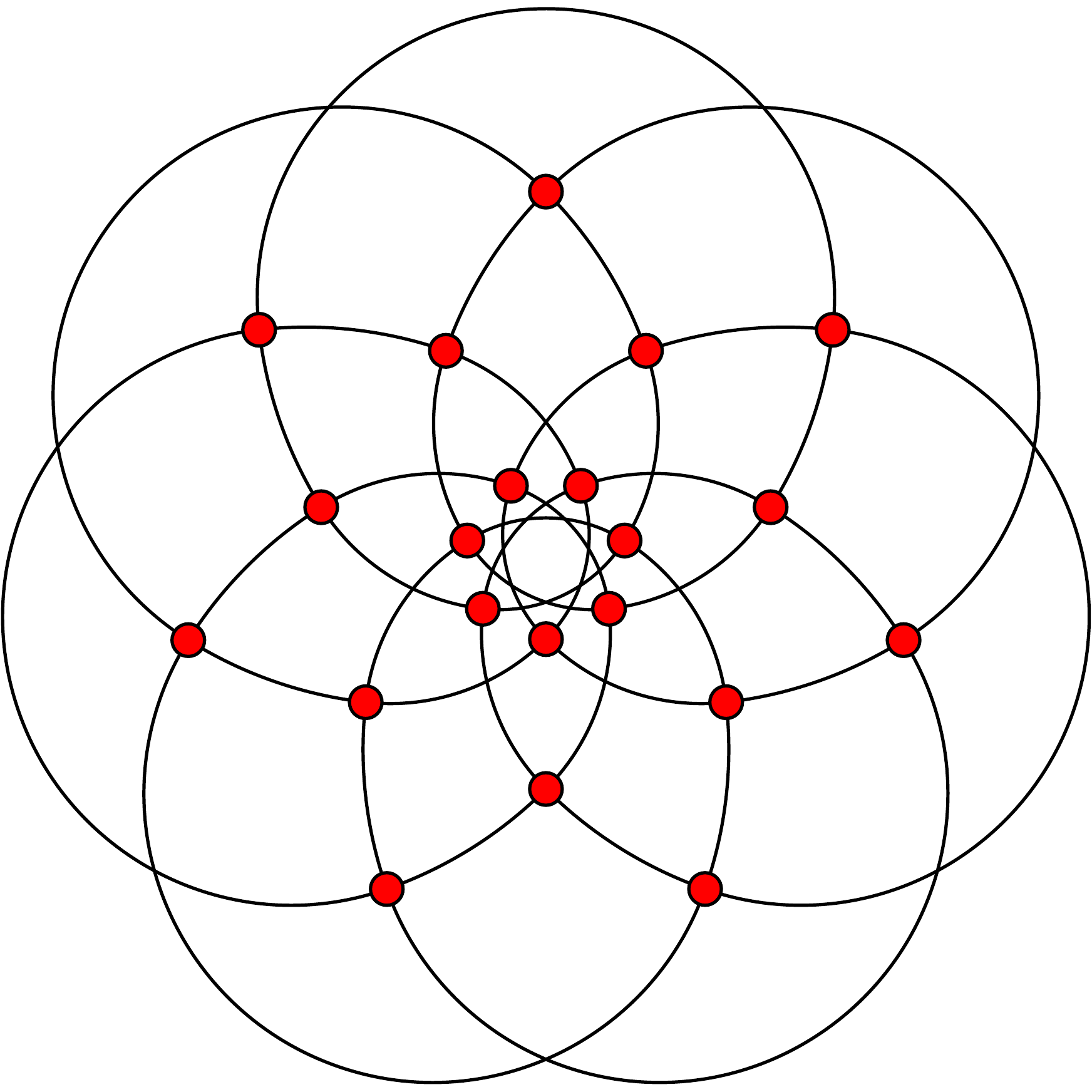}}\\
  \subfloat[Dyck graph]{\includegraphics[width=5.75cm]{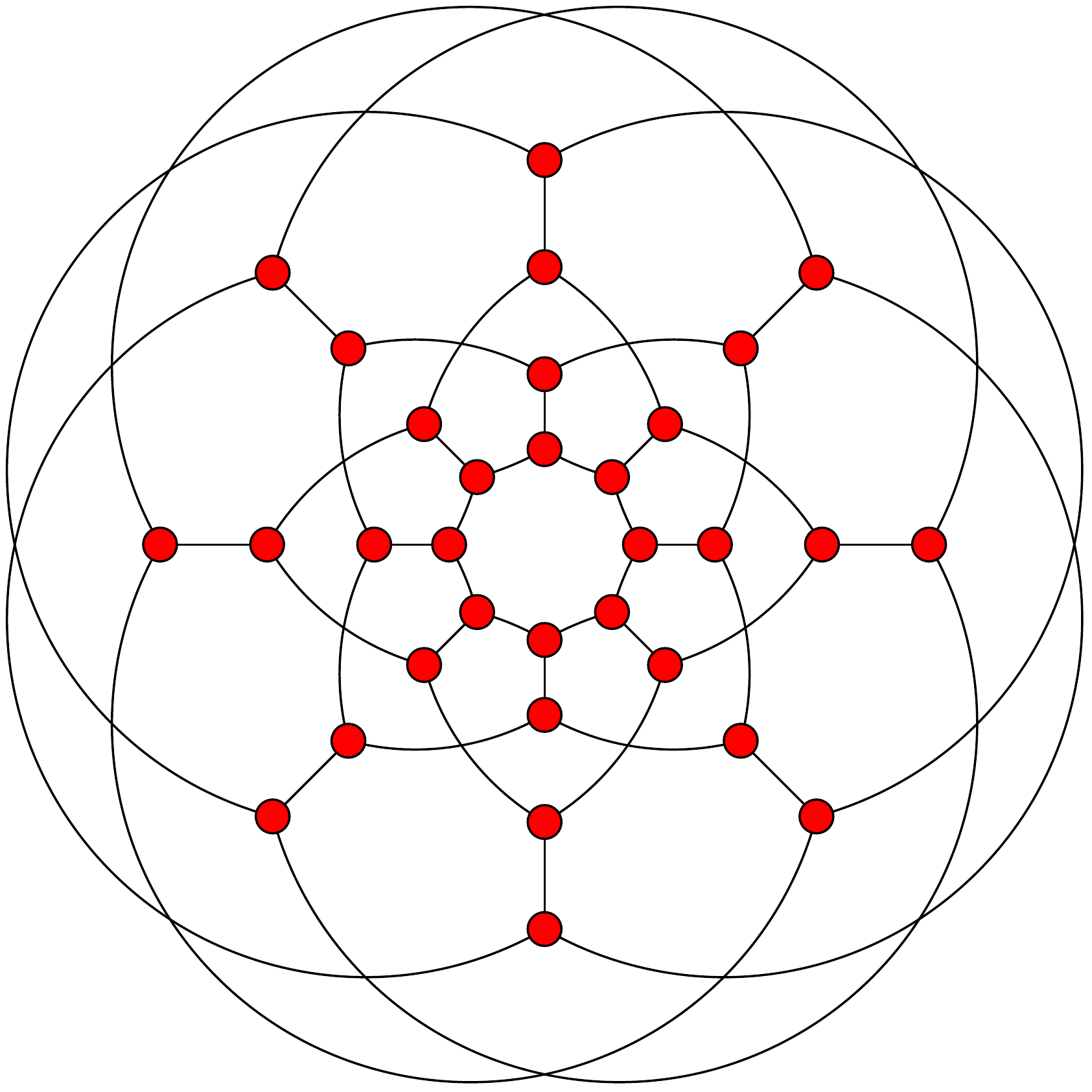}}
  \hfil
  \subfloat[40-vertex cubic symmetric graph $F_{40}$ (the bipartite double cover of the dodecahedron)]{\includegraphics[width=5.75cm]{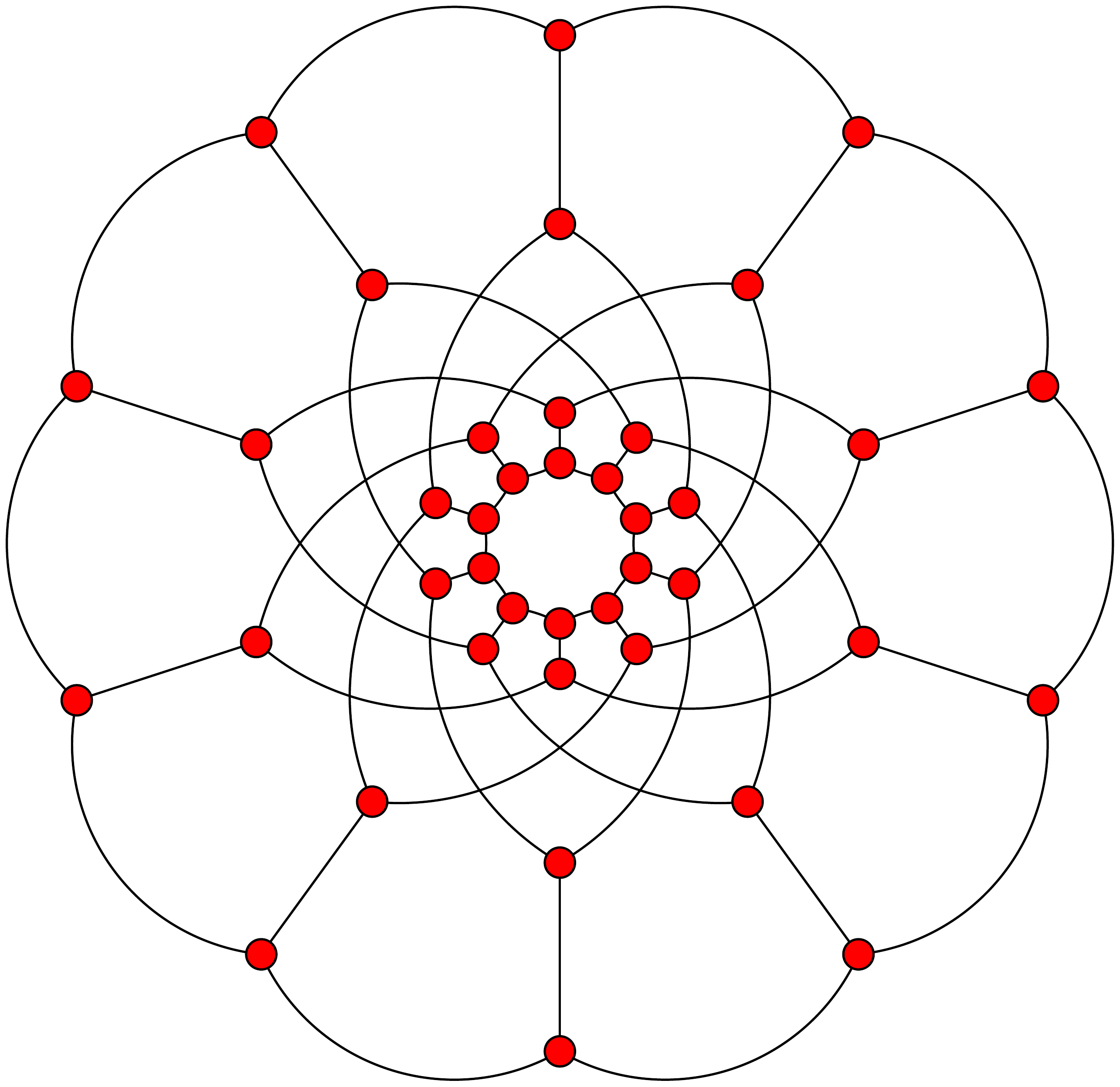}}
  \caption{Sample drawings by the Lombardi Spirograph.}
  \label{fig:spirograph}
\end{figure}

\section{Conclusions}

We have begun an investigation into Lombardi drawings
and found algorithms based on graph matching, incremental construction, 
hyperbolic geometry, 
and symmetry display for constructing drawings of this type. 
Based on our constructions, we can show that many regular graphs, 
sparse graphs, special classes of planar graphs, and symmetric 
graphs have Lombardi drawings, and we 
have found drawings of this type for many well-known graphs.
In addition, we have implemented a method,
called the \emph{Lombardi Spirograph},
for producing Lombardi drawings of graphs with dihedral symmetry.

There are many related problems that remain open, including the
following:
\begin{enumerate}
\ifProceeding
% Save a tiny dab of space
\setlength{\itemsep}{0pt}
\fi
\item What are the complexities of finding circular 
Lombardi drawings for regular graphs with degrees that are 2 mod 4? 
\item Is there an effective classification of 3-degenerate graphs according 
to whether they can or cannot be drawn in a way that avoids 
overlapping features? 
\item
Are there efficient methods for producing planar Lombardi 
drawings for outerplanar graphs, series-parallel graphs, and 
3-regular planar graphs? 
\end{enumerate}
It would also be of interest to combine Lombardi drawing with other 
standard graph drawing quality criteria such as edge-length
minimization. 
In general,
we believe that Lombardi drawings will be 
a fruitful area for much additional research.

\ifProceeding
% Save some space and bring a figure closer inline to its reference.
\vspace{-12pt}
\fi

\subsection*{Acknowledgments}
This research was supported in part by the National Science
Foundation under grant 0830403, by the
Office of Naval Research under MURI grant N00014-08-1-1015, and by the
German Research Foundation under grant NO 899/1-1.

\ifProceeding
% Save some space and bring a figure closer inline to its reference.
\vspace{-12pt}
\fi

{ %\small
\setlength{\itemsep}{0pt}
\raggedright
\bibliographystyle{abuser}
\bibliography{lombardi}
}

\ifUseAppendix
\InAppendixtrue
\DeferToAppendixfalse

\clearpage
\appendix
\section{Additional Proofs and Figures}

\fi

\end{document}